\documentclass{amsart}
\usepackage[foot]{amsaddr}
\usepackage{amsmath}
\usepackage{amssymb,amsthm}
\usepackage{mathtools}
\usepackage{pgfplots}
\pgfplotsset{compat=1.16}
\usepackage{gastex}

\usepackage{cite}

\usepackage{mathptmx}      
\theoremstyle{plain}
\newtheorem{lemma}{Lemma}
\newtheorem{theorem}{Theorem}
\newtheorem{corollary}{Corollary}
\newtheorem{proposition}{Proposition}

\theoremstyle{remark}

\newtheorem{remarknew}{Remark}

\DeclareSymbolFont{rsfscript}{OMS}{rsfs}{m}{n}
\DeclareSymbolFontAlphabet{\mathrsfs}{rsfscript}

\DeclareMathOperator{\dt}{.}

\DeclareMathOperator{\excl}{\mathrm{excl}}
\DeclareMathOperator{\dupl}{\mathrm{dupl}}
\DeclareMathOperator{\Cay}{\mathrm{Cay}}

\newcommand{\cra}{completely reachable automata}
\newcommand{\cran}{completely reachable automaton}

\newcommand{\scc}{strongly connected component}
\newcommand{\scn}{strongly connected}

\newcommand{\mA}{\mathrsfs{A}}
\newcommand{\mB}{\mathrsfs{B}}

\newcommand{\mE}{\mathrsfs{E}}
\newcommand{\Z}{\mathbb{Z}}

\usepackage{enumerate}
\usepackage{url}
\begin{document}

\title[Binary Completely Reachable Automata]{Binary completely reachable automata}

\author[D. Casas, M. V. Volkov]{David Casas \and Mikhail V. Volkov}

\address{Institute of Natural Sciences and Mathematics, Ural Federal University\\ 620000 Ekaterinburg, Russia\\
dafecato4@gmail.com, m.v.volkov@urfu.ru}

\thanks{The authors were supported by the Ministry of Science and Higher Education of the Russian Federation, project FEUZ-2020-0016.}

\begin{abstract}
A deterministic finite automaton in which every non-empty set of states occurs as the image of the whole state set under the action of a suitable input word is called completely reachable. We study completely reachable automata with two input letters.
\end{abstract}

\maketitle

\section{Introduction}
\label{sec:intro}

\emph{Completely reachable automata} are complete deterministic finite automata in which every non-empty subset of the state set occurs as the image of the whole state set under the action of a suitable input word. Such automata appeared in the study of descriptional complexity of formal languages~\cite{Maslennikova:2012,BondarVolkov16} and in relation to the \v{C}ern\'{y} conjecture~\cite{Don16}. A systematic study of \cra\ was initiated in~\cite{BondarVolkov16,BondarVolkov18} and continued in \cite{BCV:2022}. In \cite{BondarVolkov18,BCV:2022} completely reachable automata were characterized in terms of a certain finite sequence of directed graphs (digraphs): the automaton is completely reachable if and only if the final digraph in this sequence is strongly connected. In \cite[Theorem 11]{BCV:2022} it was shown that given an automaton $\mA$ with $n$ states and $m$ input letters, the $k$-th digraph in the sequence assigned to $\mA$ can be constructed in $O(mn^{2k}\log n)$ time. However, this does not yet ensure a polynomial-time algorithm for recognizing complete reachability: a series of examples in \cite{BCV:2022} demonstrates that the length of the digraph sequence for an automaton with $n$ states may reach $n-1$.

Here we study \cra\ with two input letters; for brevity, we call automata with two input letters \emph{binary}. Our main results provide a new characterization of binary \cra, and the characterization leads to a  quasilinear time algorithm for recognizing complete reachability for binary automata.

Our prerequisites are minimal: we only assume the reader's acquaintance with basic properties of strongly connected digraphs, subgroups, and cosets.

\section{Preliminaries}
\label{sec:background}

A \emph{complete deterministic finite automaton} (DFA) is a triple $\mathrsfs{A}=\langle Q,\Sigma,\delta\rangle$ where $Q$ and $\Sigma$ are finite sets called the \emph{state set} and, resp., the \emph{input alphabet} of $\mA$, and $\delta\colon Q\times\Sigma\to Q$ is a totally defined map called the \emph{transition function} of $\mA$.

The elements of $\Sigma$ are called \emph{input letters} and finite sequences of letters are called \emph{words over $\Sigma$}. The empty sequence is also treated as a word, called the \emph{empty word} and denoted $\varepsilon$. The collection of all words over $\Sigma$ is denoted $\Sigma^*$.

The transition function $\delta$ extends to a function $Q\times\Sigma^*\to Q$ (still denoted by $\delta$) via the following recursion: for every $q\in Q$, we set $\delta(q,\varepsilon)=q$  and $\delta(q,wa)=\delta(\delta(q,w),a)$ for all $w\in\Sigma^*$ and $a\in\Sigma$. Thus, every word $w\in\Sigma^*$ induces the transformation $q\mapsto\delta(q,w)$ of the set $Q$. The set $T(\mA)$ of all transformations induced this way is called the \emph{transition monoid} of $\mA$; this is the submonoid generated by the transformations $q\mapsto\delta(q,a)$, $a\in\Sigma$, in the monoid of all transformations of $Q$. A DFA $\mathrsfs{B}=\langle Q,\Theta,\zeta\rangle$ with the same state set as $\mA$ is said to be \emph{syntactically equivalent} to $\mA$ if $T(\mB)=T(\mA)$.

The function $\delta$ can be further extended to non-empty subsets of the set $Q$. Namely, for every non-empty subset $P\subseteq Q$ and every word $w\in\Sigma^*$, we let $\delta(P,w)=\{\delta(q,w)\mid q\in P\}$.

Whenever there is no risk of confusion, we tend to simplify our notation by suppressing the sign of the transition function; this means that we write $q\dt w$ for $\delta(q,w)$ and $P\dt w$ for $\delta(P,w)$ and specify a DFA as a pair $\langle Q,\Sigma\rangle$.

We say that a non-empty subset $P\subseteq Q$ is \emph{reachable} in $\mathrsfs{A}=\langle Q,\Sigma\rangle$ if $P=Q\dt w$ for some word $w\in\Sigma^*$. A DFA is called \emph{completely reachable} if every non-empty subset of its state set is reachable. Observe that complete reachability is actually a property of the transition monoid of $\mA$; hence, if a DFA $\mA$ is completely reachable, so is any DFA that is syntactically equivalent to $\mA$.

Given a DFA $\mathrsfs{A}=\langle Q,\Sigma\rangle$ and a word $w\in\Sigma^*$, the \emph{image} of $w$ is the set $Q\dt w$ and the \emph{excluded set} $\excl(w)$ of $w$ is the complement $Q{\setminus}Q\dt w$ of the image. The number $|\excl(w)|$ is called the \emph{defect} of $w$. If a word $w$ has defect~1, its excluded set consists of a unique state called the \emph{excluded state} for $w$. Further, for any $w\in\Sigma^*$, the set $\{p\in Q\mid p=q_1\dt w=q_2\dt w \ \text{ for some }\ q_1\ne q_2\}$ is called the \emph{duplicate set} of $w$ and is denoted by $\dupl(w)$. If $w$ has defect~1, its duplicate set consists of a unique state called the \emph{duplicate state} for $w$. We identify singleton sets with their elements, and therefore, for a word $w$ of defect~1, $\excl(w)$ and $\dupl(w)$ stand for its excluded and, resp., duplicate states.

For any $v\in\Sigma^*$, $q\in Q$, let $qv^{-1}=\{p\in Q\mid p\dt v=q\}$. Then for all $u,v\in\Sigma^*$,
\begin{align}
\label{eq:exclprod}\excl(uv)&=\{q\in Q\mid qv^{-1}\subseteq\excl(u)\},\\
\label{eq:duplprod}\dupl(uv)&=\{q\in Q\mid qv^{-1}\cap\dupl(u)\ne\varnothing\ \text{ or }\ |qv^{-1}{\setminus}\excl(u)|\ge 2\}.
\end{align}
The equalities \eqref{eq:exclprod} and \eqref{eq:duplprod} become clear as soon as the definitions of $\excl(\ )$ and $\dupl(\ )$ are deciphered. . Fig.~\ref{fig:product} provides a supporting illustration.
\begin{figure}[htb]
\begin{center}
\unitlength=0.9mm
\linethickness{0.4pt}
\gasset{Nw=2,Nh=2,AHLength=2,AHdist=3}
\begin{picture}(60,55.00)(0,5)
\node(A1)(5.00,5.00){}
\node(A2)(5.00,15.00){}
\node(A3)(5.00,25.00){}
\node(A4)(5.00,35.00){}
\node(A5)(5.00,45.00){}
\node(A6)(5.00,55.00){}
\node(B1)(30.00,5.00){}
\node(B2)(30.00,15.00){}
\node(B3)(30.00,25.00){}
\node(B4)(30.00,35.00){}
\node(B5)(30.00,45.00){}
\node(B6)(30.00,55.00){}
\node(C1)(55.00,5.00){}
\node(C2)(55.00,15.00){}
\node(C3)(55.00,25.00){}
\node(C4)(55.00,35.00){}
\node(C5)(55.00,45.00){}
\node(C6)(55.00,55.00){}
\put(17.00,57.00){\makebox(0,0)[cb]{$u$}}
\put(42.00,57.00){\makebox(0,0)[cb]{$v$}}
\put(56,19){$\left.\rule{0pt}{7mm}\right\}\ \dupl(uv)$}
\put(56,44){$\left.\rule{0pt}{13mm}\right\}\ \excl(uv)$}
\put(15,49){$\excl(u)\!\left\{\rule{0pt}{7mm}\right.$}
\drawedge(A1,B1){}
\drawedge[eyo=-.5](A2,B2){}
\drawedge(A3,B2){}
\drawedge[eyo=.5](A4,B2){}
\drawedge(A5,B3){}
\drawedge(A6,B4){}
\drawedge(B1,C1){}
\drawedge(B2,C2){}
\drawedge[eyo=-.5](B3,C3){}
\drawedge(B4,C3){}
\drawedge[eyo=.5](B5,C3){}
\drawedge(B6,C4){}
\end{picture}
\end{center}
\caption{An illustration for the equalities \eqref{eq:exclprod} and \eqref{eq:duplprod}}\label{fig:product}
\end{figure}
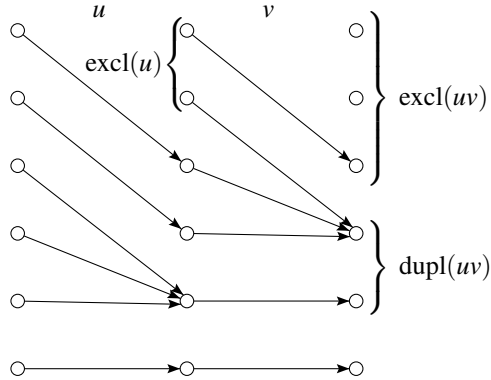

Recall that DFAs with two input letters are called \emph{binary}. The question of our study is: under which conditions is a binary DFA completely reachable? The rest of the section presents a series of reductions showing that to answer this question, it suffices to analyze DFAs of a specific form.

Let $\mathrsfs{A}=\langle Q,\{a,b\}\rangle$ be a binary DFA with $n>1$ states. If neither $a$ nor $b$ has defect~1, no subset of size $n-1$ is reachable in $\mA$. Therefore, when looking for binary \cra, we must focus on DFAs possessing a letter of defect 1. We will always assume that $a$ has defect 1.

The image of every non-empty word over $\{a,b\}$ is contained in either $Q\dt a$ or $Q\dt b$. If the defect of $b$ is greater than or equal to 1, then at most two subsets of size $n-1$ are reachable (namely, $Q\dt a$ and $Q\dt b$), whence $\mathrsfs{A}$ can only be completely reachable provided that $n=2$. The automaton $\mathrsfs{A}$ is then nothing but the classical flip-flop, see Fig.~\ref{fig:flip-flop}.
\begin{figure}[hbt]
\begin{center}
\begin{tikzpicture}
	\node[fill=white, circle, draw=blue, scale=1] (0) {$0$};
	\node[fill=white, circle, draw=blue, scale=1, right of = 0, xshift= 2cm] (1) {$1$};
	\draw
		(0) edge[-latex, bend left, above]  node{$b$} (1)
		(1) edge[-latex, bend left, below] node{$a$}(0)
		(0) edge[-latex, loop left, left, in =135, out = -135, distance = 30] node{$a$} (0)
		(1) edge[-latex, loop right, right, out=45, in = -45, distance = 30] node{$b$} (1)
		;
\end{tikzpicture}
 \caption{The flip-flop. Here and below a DFA $\langle Q,\Sigma\rangle$ is depicted as a digraph with the vertex set $Q$ and a labeled edge $q\xrightarrow{a}q'$ for each triple $(q,a,q')\in Q\times\Sigma\times Q$ such that $q\dt a=q'$.}\label{fig:flip-flop}
\end{center}
\end{figure}
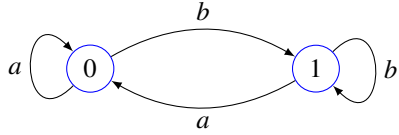

Having isolated this exception, we assume from now on that $n\ge2$ and the letter $b$ has defect 0, which means that $b$ acts as a permutation of $Q$. The following fact was first stated in~\cite{BondarVolkov16}; for a proof, see, e.g., \cite[Sect. 6]{BCV:2022}.

\begin{lemma}
\label{lem:cyclic}
If $\mathrsfs{A}=\langle Q,\{a,b\}\rangle$ is a \cran\ in which the letter $b$ acts as a permutation of $Q$, then $b$ acts as a cyclic permutation.
\end{lemma}

Taking Lemma~\ref{lem:cyclic} into account, we restrict our further considerations to DFAs with $n\ge2$ states and two input letters $a$ and $b$ such that $a$ has defect~1 and $b$ acts a cyclic permutation. Without any loss, we will additionally assume that these DFAs have the set $\mathbb{Z}_n=\{0,1,\dots,n-1\}$ of all residues modulo $n$ as their state set and the action of $b$ at any state merely adds 1 modulo $n$. Let us also agree that whenever we deal with elements of $\Z_n$, the signs $+$ and $-$ mean addition and subtraction modulo $n$, unless the contrary is explicitly specified.

Further, we will assume that $0=\excl(a)$ as it does not matter from which origin the cyclic count of the states start.

Since $b$ is a permutation, for each $k\in\Z_n$, the transformations $q\mapsto q\dt b^ka$ and $q\mapsto q\dt b$ generate the same submonoid in the monoid of all transformations of $\Z_n$ as do the transformations $q\mapsto q\dt a$ and $q\mapsto q\dt b$. This means that if one treats the word $b^ka$ as a new letter $a_{k}$, say, one gets the DFA $\mA_{k}=\langle \Z_n,\{a_{k},b\}\rangle$ that is syntactically equivalent to $\mA$. Therefore, $\mA$ is completely reachable if and only if so is $\mA_{k}$ for some (and hence for all) $k$. Hence we may choose $k$ as we wish and study the DFA $\mA_{k}$ for the specified value of $k$ instead of $\mA$.

What can we achieve using this? From \eqref{eq:exclprod} we have $\excl(b^ka)=\excl(a)=0$. Further, let $q_1\ne q_2$ be such that $q_1\dt a=q_2\dt a=\dupl(a)$. Choosing $k=q_1$ (or $k=q_2$), we get $0\dt b^ka=\dupl(a)$. Thus, we will assume that $0\dt a=\dupl(a)$.

Summarizing, we will consider DFAs $\langle \Z_n,\{a,b\}\rangle$ such that:
\begin{itemize}
  \item the letter $a$ has defect 1, $\excl(a)=0$, and $0\dt a=\dupl(a)$;
  \item $q\dt b=q+1$ for each $q\in\Z_n$.
\end{itemize}
We call such DFAs \emph{standardized}. For the purpose of complexity considerations at the end of Sect.~\ref{sec:construction}, observe that given a binary DFA $\mA$ in which one letter acts as a cyclic permutation while the other has defect 1, one can `standardize' the automaton, that is, construct a standardized DFA syntactically equivalent to $\mA$, in linear time with respect to the size of $\mA$.

\section{A necessary condition}
\label{sec:necessary}

Let $\langle \Z_n,\{a,b\}\rangle$ be a standardized DFA and $w\in\{a,b\}^*$. A subset $S\subseteq\Z_n$ is said to be $w$-\emph{invariant} if $S\dt w\subseteq S$.

\begin{proposition}
\label{prop:invsubgroup}
If $\langle \Z_n,\{a,b\}\rangle$ is a completely reachable standardized DFA, then no proper subgroup of $(\Z_n,+)$ is $a$-invariant.
\end{proposition}

\begin{proof}
Arguing by contradiction, assume that $H\subsetneqq\Z_n$ is a subgroup such that $H\dt a\subseteq H$. Let $d$ stand for the index of the subgroup $H$ in the group $(\Z_n,+)$. The set $\Z_n$ is then partitioned into the $d$ cosets
$$H_0=H,\ H_1=H\dt b=H+1,\ \dots,\ H_{d-1}=H\dt b^{d-1}=H+d-1.$$
For $i=0,1,\dots,d-1$, let $T_i$ be the complement of the coset $H_i$ in $\Z_n$. Then we have $T_i=\cup_{j\ne i}H_j$ and $T_i\dt b=T_{i+1\!\!\pmod{d}}$ for each $i=0,1,\dots,d-1$.

Since $\mA$ is completely reachable, each subset $T_i$ is reachable. Take a word $w$ of minimum length among words with the image equal to one of the subsets $T_0,T_1,\dots,T_{d-1}$. Write $w$ as $w=w'c$ for some letter $c\in\{a,b\}$.

If $c=b$, then for some $i\in\{0,1,\dots,d-1\}$, we have
$$\Z_n\dt w'b=T_i=T_{i-1\!\!\!\!\pmod{d}}\dt b.$$
Since $b^n$ acts as the identity mapping, applying the word $b^{n-1}$ to this equality yields $\Z_n\dt w'=T_{i-1\!\!\pmod{d}}$ whence the image of $w'$ is also equal to one of the subsets $T_0,T_1,\dots,T_{d-1}$. This contradicts the choice of $w$.

Thus, $c=a$, whence the set $\Z_n\dt w$  is contained in $\Z_n\dt a$. The only $T_i$ that is contained in $\Z_n\dt a$ is $T_0$ because each $T_i$ with $i\ne 0$ contains $H_0$, and $H_0=H$ contains 0, the excluded state of $a$. Hence, $\Z_n\dt w=T_0$, that is, $\Z_n\dt w'a=T_0$. For each state $q\in\Z_n\dt w'$, we have $q\dt a\in T_0$, and this implies $q\in T_0$ since $H_0$, the complement of $T_0$, is $a$-invariant. We see that $\Z_n\dt w'\subseteq T_0$ and the inclusion cannot be strict because $T_0$ cannot be the image of its proper subset. However, the equality $\Z_n\dt w'=T_0$ again contradicts the choice of $w$.
\end{proof}

We will show that the condition of Proposition~\ref{prop:invsubgroup} is not only necessary but also sufficient for complete reachability of a standardized DFA. The proof of sufficiency requires a construction that we present in full in Sect.~\ref{sec:construction}, after studying its simplest case in Sect.~\ref{sec:sufficient}.

\section{Rystsov's graph of a binary DFA}
\label{sec:sufficient}

Recall a sufficient condition for complete reachability from~\cite{BondarVolkov16}. Given a (not necessarily binary) DFA $\mathrsfs{A}=\langle Q,\Sigma\rangle$, let $W_1(\mathrsfs{A})$ stand for the set of all words in $\Sigma^*$ that have defect~1 in $\mathrsfs{A}$. Consider a digraph with the vertex set $Q$ and the edge set 
\[
E=\{(\excl(w),\dupl(w))\mid w\in W_1(\mathrsfs{A})\}.
\] 
We denote this digraph by $\Gamma_1(\mathrsfs{A})$. The notation comes from~\cite{BondarVolkov16}, but much earlier, though in a less explicit form, the construction was used by Rystsov~\cite{rystsov2000estimation} for some special species of DFAs. Taking this into account, we refer to $\Gamma_1(\mathrsfs{A})$ as the \emph{Rystsov graph} of $\mA$.

\begin{theorem}[\!\!{\mdseries\cite[Theorem 1]{BondarVolkov16}}]
\label{thm:sufficient}
If a DFA $\mathrsfs{A}=\langle Q,\Sigma\rangle$ is such that the graph $\Gamma_1(\mathrsfs{A})$ is \scn, then $\mathrsfs{A}$ is completely reachable.
\end{theorem}

It was shown in~\cite{BondarVolkov16} that the condition of Theorem~\ref{thm:sufficient} is not necessary for complete reachability, but it was conjectured that the condition might characterize binary \cra. However, this conjecture has been refuted in~\cite[Example~2]{BCV:2022} by exhibiting a binary \cran\ with 12 states whose Rystsov graph is not \scn. Here we include a similar example which we will use to illustrate some of our results.

Consider the standardized DFA $\mE'_{12}=\langle \Z_{12},\{a,b\}\rangle$ where the action of the letter $a$ is specified as follows:
\begin{center}
\begin{tabular}{c@{\ \  }|@{\ \  }c@{\ \  }c@{\ \  }c@{\ \  }c@{\ \  }c@{\ \  }c@{\ \  }c@{\ \  }c@{\ \  }c@{\ \  }c@{\ \  }c@{\ \  }c}
$q$      & 0  & 1 & 2 & 3 & 4  & 5  & 6  & 7 & 8 & 9  & 10 & 11 \mathstrut\\
\hline
$q\dt a$ & 10 & 1 & 2 & 8 & 4  & 5  & 10 & 9 & 3 & 7  &  6 & 11 \mathstrut
\end{tabular}\,.
\end{center}
(The DFA $\mE'_{12}$ only slightly differs from the DFA $\mE_{12}$ used in~\cite[Example~2]{BCV:2022}, hence the notation.) The DFA  $\mathrsfs{E}'_{12}$ is shown in Fig.~\ref{fig:e12}, in which we have replaced edges that should have been labeled $a$ and $b$ with solid and, resp., dashed edges.
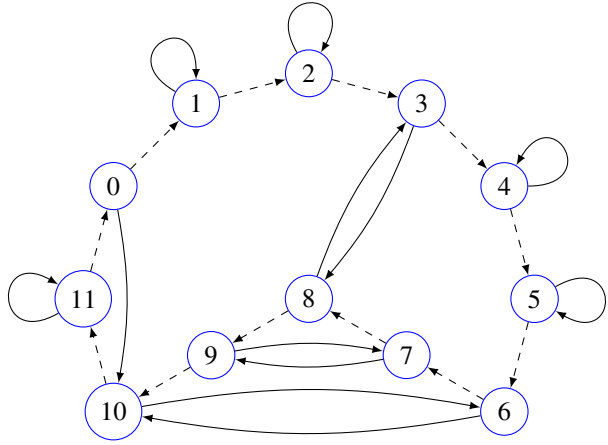
\begin{figure}[bht]
\begin{center}
\begin{tikzpicture}
 \node[fill=white, circle, draw=blue, scale=1] at (0:0cm) (8) {$8$};
 \node[fill=white, circle, draw=blue, scale=1] at (0:3cm) (5) {$5$};
 \node[fill=white, circle, draw=blue, scale=1] at (-30:3cm) (6) {$6$};
 \node[fill=white, circle, draw=blue, scale=1] at (30:3cm) (4) {$4$};
 \node[fill=white, circle, draw=blue, scale=1] at (60:3cm) (3) {$3$};
 \node[fill=white, circle, draw=blue, scale=1] at (90:3cm) (2) {$2$};
 \node[fill=white, circle, draw=blue, scale=1] at (120:3cm) (1) {$1$};
 \node[fill=white, circle, draw=blue, scale=1] at (150:3cm) (0) {$0$};
 \node[fill=white, circle, draw=blue, scale=1] at (180:3cm) (11) {$11$};
 \node[fill=white, circle, draw=blue, scale=1] at (210:3cm) (10) {$10$};
 \node[fill=white, circle, draw=blue, scale=1] at (-30:1.5cm) (7) {$7$};
 \node[fill=white, circle, draw=blue, scale=1] at (210:1.5cm) (9) {$9$};

 \draw
 (0) edge[-latex, dashed] (1)
 (1) edge[-latex, dashed] (2)
 (2) edge[-latex, dashed] (3)
 (3) edge[-latex, dashed] (4)
 (4) edge[-latex, dashed] (5)
 (5) edge[-latex, dashed] (6)
 (6) edge[-latex, dashed] (7)
 (7) edge[-latex, dashed] (8)
 (8) edge[-latex, dashed] (9)
 (9) edge[-latex, dashed] (10)
 (10) edge[-latex, dashed] (11)
 (11) edge[-latex, dashed] (0)

 (0) edge[-latex, bend left=10] (10)
 (8) edge[-latex, bend left=10] (3)
 (3) edge[-latex, bend left=10] (8)
 (9) edge[-latex, bend left=10] (7)
 (7) edge[-latex, bend left=10] (9)
 (10) edge[-latex, bend left=10] (6)
 (6) edge[-latex, bend left=10] (10)

 (11) edge[-latex, out = 210, in=150, distance = 1cm](11)
 (1) edge[-latex, out=150, in=90, distance=1cm ](1)
 (2) edge[-latex, out=120, in = 60, distance = 1cm ](2)
 (4) edge[-latex, loop right, out = 0, in=60, distance=1cm ](4)
 (5) edge[-latex, out = 30, in = -30, distance = 1cm ](5)

 ;
\end{tikzpicture}
\caption{The DFA $\mathrsfs{E}'_{12}$; solid and dashed edges show the action of $a$ and, resp., $b$}\label{fig:e12}
\end{center}
\end{figure}

We postpone the description of the digraph $\Gamma_1(\mathrsfs{E}'_{12})$ and the proof that the DFA $\mathrsfs{E}'_{12}$ is completely reachable until we develop suitable tools that make the description and the proof easy.

We start with a characterization of Rystsov's graphs of standardized DFAs. Let $\mA=\langle \Z_n,\{a,b\}\rangle$ be such a DFA. It readily follows from \eqref{eq:exclprod} and \eqref{eq:duplprod} that $\excl(w)\dt b=\excl(wb)$ and $\dupl(w)\dt b=\dupl(wb)$ for every word $w\in W_1(\mA)$. Therefore, the edge set $E$ of the digraph $\Gamma_1(\mA)$ is closed under the \emph{translation} $(q,p)\mapsto (q\dt b, p\dt b)=(q+1,p+1)$.
As a consequence, for any edge $(q,p)\in E$ and any $k$, the pair $(q+k,p+k)$ also constitutes an edge in $E$.

Denote by $D_1(\mA)$ the set of ends of edges of $\Gamma_1(\mA)$ that start at 0, that is, $D_1(\mA)=\{p\in\Z_n\mid (0,p)\in E\}$. We call $D_1(\mA)$ the \emph{difference set} of $\mA$. Our first observation shows how to recover all edges of $\Gamma_1(\mA)$, knowing $D_1(\mA)$.
\begin{lemma}
\label{lem:difference}
Let $\mA=\langle \Z_n,\{a,b\}\rangle$ be a standardized DFA. A pair $(q,p)\in\Z_n\times \Z_n$ forms an edge in the digraph $\Gamma_1(\mA)$ if and only if $p-q\in D_1(\mA)$.
\end{lemma}

\begin{proof}
If $p-q\in D_1(\mA)$, the pair $(0,p-q)$ is an edge in $E$, and therefore, so is the pair $(0+q,(p-q)+q)=(q,p)$. Conversely, if $(q,p)$ is an edge in $E$, then so is $(q+(n-q),p+(n-q))=(0,p-q)$, whence $p-q\in D_1(\mA)$.
\end{proof}

By Lemma~\ref{lem:difference}, the presence or absence of an edge in  $\Gamma_1(\mA)$ depends only on the difference modulo $n$ of two vertex numbers. This means that $\Gamma_1(\mA)$ is a \emph{circulant} digraph, that is, the Cayley digraph of the cyclic group $(\Z_n,+)$ with respect to some subset of $\Z_n$. Recall that if $D$ is a subset in a group $G$, the \emph{Cayley digraph of $G$ with respect to $D$}, denoted $\Cay(G,D)$, has $G$ as its vertex set and $\{(g,gd)\mid g\in G,\ d\in D\}$ as its edge set. The following property of Cayley digraphs of finite groups is folklore\footnote{In fact, our definition is the semigroup version of the notion of a Cayley digraph, but this makes no difference since in a finite group, every subsemigroup is a subgroup.}.
\begin{lemma}
\label{lem:cayley}
Let $G$ be a finite group, $D$ a subset of $G$, and $H$ the subgroup of $G$ generated by $D$. The \scc{}s\ of the Cayley digraph\/ $\Cay(G,D)$ have the right cosets $Hg$, $g\in G$, as their vertex sets, and each \scc\ is isomorphic to $\Cay(H,D)$. In particular, the digraph $\Cay(G,D)$ is \scn\ if and only if $G$ is generated by $D$.
\end{lemma}

Let $H_1(\mA)$ stand for the subgroup of the group $(\Z_n,+)$ generated by the difference set $D_1(\mA)$. Specializing Lemma~\ref{lem:cayley}, we get the following description for Rystsov's graphs of standardized DFAs.
\begin{proposition}
\label{prop:cayley}
Let $\mA=\langle \Z_n,\{a,b\}\rangle$ be a standardized DFA. The digraph $\Gamma_1(\mA)$ is isomorphic to the Cayley digraph $\Cay(\Z_n,D_1(\mA))$.  The \scc{}s\ of $\Gamma_1(\mA)$ have the cosets of the subgroup $H_1(\mA)$ as their vertex sets, and each \scc\ is isomorphic to the Cayley digraph $\Cay(H_1(\mA),D_1(\mA))$. In particular, the digraph $\Gamma_1(\mA)$ is \scn\ if and only if the set $D_1(\mA)$ generates $(\Z_n,+)$ or, equivalently, if and only if the greatest common divisor of $D_1(\mA)$ is coprime to $n$.
\end{proposition}

Proposition~\ref{prop:cayley} shows that structure of the Rystsov graph of a standardized DFA $\mA$ crucially depends on its difference set $D_1(\mA)$. The definition of the edge set of $\Gamma_1(\mA)$ describes $D_1(\mA)$ as the set of duplicate states for all words $w$ of defect 1 whose excluded state is 0, that is, $D_1(\mA)=\{\dupl(w)\mid \excl(w)=0\}$. Thus, understanding of difference sets amounts to a classification of transformations caused by words of defect 1. It is such a classification that is behind the following handy description of difference sets.

\begin{proposition}
\label{prop:d0}
Let $\mA=\langle \Z_n,\{a,b\}\rangle$ be a standardized DFA. Let $r\ne 0$ be such that $r\dt a=\dupl(a)$. Then
\begin{equation}
\label{eq:d0}
D_1(\mA)=\{\dupl(a)\dt v \mid v\in\{a,b^ra\}^*\}.
\end{equation}
\end{proposition}

\begin{proof}
Denote by $N$ the image of the letter $a$, that is, $N=\Z_n{\setminus}\{0\}$. If $q\dt a=p$ for some $q\in\Z_n$ and $p\in N$, then, clearly, $(q-r)\dt b^ra=p$. Hence the only state in $N$ that has a preimage of size 2 under the actions of both $a$ and $b^ra$ is
\[
\dupl(a)=\begin{cases}
 0\dt a=r\dt a,\\
 (n-r)\dt b^ra=0\dt b^ra,
 \end{cases}
\]
and in both cases 0 belongs to the preimage. Thus, the preimage of every $p\in N$ under both $a$ and $b^ra$ contains a unique state in $N$, which means that both $a$ and $b^ra$ act on the set $N$ as permutations. Hence every word $v\in\{a,b^ra\}^*$ acts on $N$ as a permutation. Then the word $av$ has defect 1 and $\excl(av)=0$. Applying the equality \eqref{eq:duplprod} with $a$ in the role of $u$, we derive that $\dupl(av)=\dupl(a)\dt v$. Thus, denoting the right-hand side of \eqref{eq:d0} by $D$, we see that every state in $D$ is the duplicate state of some word whose only excluded state is 0. This means that $D_1(\mA)\supseteq D$.

To verify the converse inclusion, take an arbitrary state $p\in D_1(\mA)$ and let $w$ be a word of defect 1 such that $\excl(w)=0$ and $\dupl(w)=p$. Since $\excl(w)=0$, the word $w$ ends with the letter $a$. We prove that $p$ lies in $D$ by induction on the number of occurrences of $a$ in $w$. If $a$ occurs in $w$ once, then $w=b^ka$ for some $k\in\Z_n$. We have $p=\dupl(w)=\dupl(b^ka)=\dupl(a)\in D$.

If $a$ occurs in $w$ at least twice, write $w=w'b^ka$ where $w'$ ends with $a$. Then the word $w'$ has defect 1 and $\excl(w')=0$. As $w'$ has fewer occurrences of $a$, the inductive assumption applies and yields $\dupl(w')\in D$. Denoting $\dupl(w')$ by $p'$, we have $p=p'\dt b^ka$. If we prove that $k\in\{0,r\}$, we are done since the set $D$ is both $a$-invariant and $b^ra$-invariant by its definition. Arguing by contradiction, assume $k\notin\{0,r\}$. Let $\ell=k\dt a$; then $k$ is the only state in $\ell a^{-1}$. Hence $\ell a^{-1}=\excl(w'b^{k})$, and the equality \eqref{eq:exclprod} (with $u=w'b^{k}$ and $v=a$) shows that $\ell\in\excl(w'b^{k}a)=\excl(w)$. Clearly, $\ell\ne 0$ as $\ell$ lies in the image of $a$. Therefore the conclusion $\ell\in\excl(w)$ contradicts the assumption $\excl(w)=0$.
\end{proof}

For an illustration, we apply \eqref{eq:d0} to compute the difference set for the DFA $\mathrsfs{E}'_{12}$ shown in Fig.~\ref{fig:e12}. In $\mathrsfs{E}'_{12}$, we have $r=6$ and $\dupl(a)=10$. Acting by $a$ and $b^6a$ gives $10\dt a=6$ and $10\dt b^6a=(10+6)\dt a=4\dt a=4$. Thus, $4,6\in D_1(\mathrsfs{E}'_{12})$. Acting by $a$ or $b^6a$ at 4 and 6 does not produce anything new: $4\dt a=4$ and $4\dt b^6a=(4+6)\dt a=10\dt a=6$ while $6\dt a=10$ and $6\dt b^6a=(6+6)\dt a=0\dt a=10$. We conclude that $D_1(\mathrsfs{E}'_{12})=\{4,6,10\}$. Since 2, the greatest common divisor of $\{4,6,10\}$, divides 12, we see that the digraph $\Gamma_1(\mathrsfs{E}'_{12})$ is not \scn. The subgroup $H_1(\mathrsfs{E}'_{12})$ consists of even residues modulo 12 and has index 2. Hence the digraph $\Gamma_1(\mathrsfs{E}'_{12})$ has two \scc{}s whose vertex sets are $\{0,2,4,6,8,10\}$ and $\{1,3,5,7,9,11\}$,
and for each $q\in\Z_{12}$, it has the edges $(q,q+4)$, $(q,q+6)$, and $(q,q+10)$.

\medskip

In fact, formula \eqref{eq:d0} leads to a straightforward algorithm that computes the difference set of any standardized DFA $\mA$ in time linear in $n$. This, together with Proposition~\ref{prop:cayley}, gives an efficient way to compute the Rystsov graph of $\mA$.

\medskip

Let $D_1^0(\mA)=D_1(\mA)\cup\{0\}$. It turns out that $D_1^0(\mA)$ is always a union of cosets of a nontrivial subgroup.
\begin{proposition}
\label{prop:d0coset}
Let $\mA=\langle \Z_n,\{a,b\}\rangle$ be a standardized DFA. Let $r\ne 0$ be such that $r\dt a=\dupl(a)$. Then the set $D_1^0(\mA)$ is a union of cosets of the subgroup generated by $r$ in the group $H_1(\mA)$.
\end{proposition}

\begin{proof}
It is easy to see that the claim is equivalent to the following implication: if $d\in D_1^0(\mA)$, then $d+r\in D_1^0(\mA)$. This clearly holds if $d+r=0$. Thus, assume that $d\in D_1^0(\mA)$ is such that $d+r\ne0$. Then $(d+r)\dt a\in D_1(\mA)$. Indeed, if $d=0$, then $(d+r)\dt a=r\dt a=\dupl(a)\in D_1(\mA)$. If $d\ne 0$, then $d\in D_1(\mA)$, whence $(d+r)\dt a=d\dt b^ra\in D_1(\mA)$ as formula \eqref{eq:d0} ensures that the set $D_1(\mA)$ is closed under the action of the word $b^ra$.

We have observed in the first paragraph of the proof of Proposition~\ref{prop:d0} that $a$ acts on the set $N=\Z_n{\setminus}\{0\}$ as a permutation. Hence for some $k$, the word $a^k$ acts on $N$ as the identity map. Then
$d+r=(d+r)\dt a^k=((d+r)\dt a)\dt a^{k-1}\in D_1(\mA)$ since we have already shown that $(d+r)\dt a\in D_1(\mA)$ and formula \eqref{eq:d0} ensures that the set $D_1(\mA)$ is $a$-invariant.
\end{proof}

In our running example $\mathrsfs{E}'_{12}$, $r=6$ and the set $D_1^0(\mathrsfs{E}'_{12})=\{0,4,6,10\}$ is the union of the subgroup $\{0,6\}$ with its coset $\{4,10\}$ in the group $H_1(\mathrsfs{E}'_{12})$.

\medskip

Let $\mA=\langle \Z_n,\{a,b\}\rangle$ be a standardized DFA. Proposition~\ref{prop:d0coset} shows that then the set $D_1^0(\mA)$ is situated between the subgroup $H_1(\mA)$ and the subgroup $R$ generated by $r\ne 0$ such that $r\dt a=\dupl(a)$:
\begin{equation}
\label{eq:inclusions}
R\subseteq D_1^0(\mA)\subseteq H_1(\mA).
\end{equation}
Formula \eqref{eq:d0} implies that the difference set $D_1(\mA)$ is $a$-invariant, and so is the set $D_1^0(\mA)$ since $0\dt a=\dupl(a)\in D_1(\mA)$. By Proposition~\ref{prop:invsubgroup}, if the automaton $\mA$ is completely reachable, then either $H_1(\mA)=\Z_n$ or $H_1(\mA)$ is a proper subgroup and both inclusions in \eqref{eq:inclusions} are strict. Recall that by Proposition~\ref{prop:cayley} $H_1(\mA)=\Z_n$ if and only if the digraph $\Gamma_1(\mA)$ is \scn. In the other case, $n$ must be a product of at least three (not necessarily distinct) prime numbers. Indeed, the subgroups of $(\Z_n,+)$ ordered by inclusion are in a 1-1 correspondence to the divisors of $n$ ordered by division, and no product of only two primes can have two different proper divisors $d_1$ and $d_2$ such that $d_1$ divides $d_2$. We thus arrive at the following conclusion.

\begin{corollary}
\label{cor:2primes}
A binary DFA $\mA$ with $n$ states where $n$ is a product of two prime numbers is completely reachable if and only if one of its letters acts as a cyclic permutation of the state set, the other letter has defect $1$, and the digraph $\Gamma_1(\mA)$ is \scn.
\end{corollary}

Corollary~\ref{cor:2primes} allows one to show that the number of states in a binary \cra\ whose Rystsov graph is not \scn\ is at least 12. (Thus, our examples of such automata ($\mE_{12}$ from \cite[Example~2]{BCV:2022} and $\mE'_{12}$ from the present paper) are of minimum possible size.) Indeed, Corollary~\ref{cor:2primes} excludes all sizes less than 12 except 8. If a standardized DFA $\mA$ has 8 states and the digraph $\Gamma_1(\mA)$ is not \scn, then the group $H_1(\mA)$ has size at most 4 and its subgroup $R$ generated by the non-zero state in $\dupl(a)a^{-1}$ has size at least 2. By Proposition~\ref{prop:d0coset} the set $D_1^0(\mA)$ is a union of cosets of the subgroup $R$ in the group $H_1(\mA)$, whence either $D_0(\mA)=R$ or $D_0(\mA)=H_1(\mA)$. In either case, we get a proper $a$-invariant subgroup, and Proposition~\ref{prop:invsubgroup} implies that the DFA $\mA$ is not completely reachable.

\section{Subgroup sequences for standardized DFAs}
\label{sec:construction}

In \cite{BondarVolkov18,BCV:2022} Theorem~\ref{thm:sufficient} is generalized in the following way. A sequence of digraphs $\Gamma_1(\mA)$, $\Gamma_2(\mA)$, \dots, $\Gamma_k(\mA)$, \dots\  is assigned to an arbitrary (not necessarily binary) DFA $\mathrsfs{A}$, where $\Gamma_1(\mA)$ is the Rystsov graph of $\mA$ while the `higher level' digraphs $\Gamma_2(\mA)$, \dots, $\Gamma_k(\mA)$, \dots\ are defined via words that have defect~2, \dots, $k$, \dots\ in $\mathrsfs{A}$. (We refer the interested reader to \cite{BondarVolkov18,BCV:2022} for the precise definitions; here we do not need them.) The length of the sequence is less than the number of states of $\mA$, and $\mA$ is completely reachable if and only if the final  digraph in the sequence is strongly connected.

For the case when $\mA$ is a standardized DFA, Proposition~\ref{prop:cayley} shows that the Rystsov graph $\Gamma_1(\mA)$ is completely determined by the difference set $D_1(\mA)$ and the subgroup $H_1(\mA)$ that $D_1(\mA)$  generates. This suggests that for binary automata, one may substitute the `higher level' digraphs of \cite{BondarVolkov18,BCV:2022} by suitably chosen `higher level' difference sets and their generated subgroups.

Take a standardized DFA $\mA=\langle \Z_n,\{a,b\}\rangle$ and for each $k>1$, inductively define the set $D_k(\mA)$ and the subgroup $H_k(\mA)$:
\begin{align}
\label{eq:defk}
D_k(\mA)&=\{p\in\Z_n\mid p\in\dupl(w) \text{ for some } w\in\{a,b\}^*\notag\\
        &\text{ such that } 0\in\excl(w)\subseteq H_{k-1}(\mA),\ |\excl(w)|\le k\},\\
H_k(\mA)&\text{ is the subgroup of $(\Z_n,+)$ generated by } D_k(\mA).\notag
\end{align}
Observe that if we let $H_0(\mA)=\{0\}$, the definition \eqref{eq:defk} makes sense also for $k=1$ and leads to exactly the same $D_1(\mA)$ and $H_1(\mA)$ as defined in Sect.~\ref{sec:sufficient}.

Using the definition \eqref{eq:defk}, it is easy to prove by induction that $D_k(\mA)\subseteq D_{k+1}(\mA)$ and $H_k(\mA)\subseteq H_{k+1}(\mA)$ for all $k$.

\begin{proposition}
\label{prop:sufficientnew}
If $\mA=\langle \Z_n,\{a,b\}\rangle$ is a standardized DFA and $H_\ell(\mA)=\Z_n$ for some $\ell$, then $\mA$ is a \cran.
\end{proposition}

\begin{proof}
As $\mA$ is fixed, we write $D_k$ and $H_k$ instead of $D_k(\mA)$ and, resp., $H_k(\mA)$.

Take any non-empty subset $S\subseteq\Z_n$. We prove that $S$ is reachable in $\mathrsfs{A}$ by induction on $n-|S|$. If $n-|S|=0$, there is nothing to prove as $S=\Z_n$ is reachable via the empty word. Now let $S$ be a proper subset of $\Z_n$. We aim to find a subset $T\subseteq\Z_n$ such that $S=T\dt v$ for some word $v\in\{a,b\}^*$ and $|T|>|S|$. Since $n-|T|<n-|S|$, the induction assumption applies to the subset $T$ whence $T=\Z_n\dt u$ for some word $u\in\{a,b\}^*$. Then $S=\Z_n\dt uv$ is reachable as required.

Thus, fix a non-empty subset $S\subsetneqq\Z_n$. Since cosets of the trivial subgroup $H_0$ are singletons, $S$ is a union of cosets of $H_0$. On the other hand, since $H_\ell=\Z_n$, the only coset of $H_\ell$ strictly contains $S$, and so $S$ is not a union of cosets of $H_\ell$. Now choose $k\ge1$ to be the maximal number for which $S$ is a union of cosets of the subgroup $H_{k-1}$. The subgroup $H_k$ already has a coset, say, $H_k+t$ being neither contained in $S$ nor disjoint with $S$; in other words, $\varnothing\ne S\cap (H_k+t)\subsetneqq H_k+t$.

By Lemma~\ref{lem:cayley}, the coset $H_k+t$ serves as the vertex set of a \scc\ of the Cayley digraph $\Cay(\Z_n,D_k)$. Therefore, some edge of $\Cay(\Z_n,D_k)$ connects $(H_k+t)\setminus S$ with $S\cap(H_k+t)$ in this \scc, that is, the head $q$ of this edge lies in $(H_k+t)\setminus S$ while its tail $p$ belongs to $S\cap(H_k+t)$. Let $p'=p-q$; then $p'\in D_k$ by the definition of the Cayley digraph. By \eqref{eq:defk} there exists a word $w\in\{a,b\}^*$ such that $p'\in\dupl(w)$ and $\excl(w)\subseteq H_{k-1}$. Then $p=p'+q=p'\dt b^q\in\dupl(w)\dt b^q=\dupl(wb^q)$ and $\excl(wb^q)=\excl(w)\dt b^q=\excl(w)+q\subseteq H_{k-1}+q$. From $p\in \dupl(wb^q)$ we conclude that there exist $p_1,p_2\in\Z_n$ such that $p=p_1\dt wb^q=p_2\dt wb^q$. Since $S$ is a union of cosets of the subgroup $H_{k-1}$, the fact that $q\notin S$ implies that the whole coset $H_{k-1}+q$ is disjoint with $S$, and the inclusion $\excl(wb^q)\subseteq H_{k-1}+q$ ensures that $S$ is disjoint with $\excl(wb^q)$. Therefore, for every $s\in S\setminus\{p\}$, there exists a state $s'\in\Z_n$ such that $s'\dt wb^q=s$. Now letting $T=\{p_1,p_2\}\cup\bigl\{s'\mid s\in S{\setminus}\{p\}\bigr\}$, we conclude that $S=T\dt wb^q$ and $|T|=|S|+1$.
\end{proof}

For an illustration, return one last time to the DFA $\mathrsfs{E}'_{12}$ shown in Fig.~\ref{fig:e12}. We have seen that the subgroup $H_1(\mathrsfs{E}'_{12})$ consists of even residues modulo 12. Inspecting the word $ab^3a$ gives $\excl(ab^3a)=\{0,8\}\subseteq H_1(\mathrsfs{E}'_{12})$ and $1\in\dupl(ab^3a)$, whence $1\in D_2(\mathrsfs{E}'_{12})$. Therefore the subgroup $H_2(\mathrsfs{E}'_{12})$ generated by $D_2(\mathrsfs{E}'_{12})$ is equal to $\Z_{12}$, and $\mathrsfs{E}'_{12}$ is a \cran\ by Proposition~\ref{prop:sufficientnew}.

To illustrate the next level of the construction \eqref{eq:defk}, consider the standardized DFA $\mE_{48}=\langle \Z_{48},\{a,b\}\rangle$ shown in Fig.~\ref{fig:e48}. We have replaced edges that should have been labeled $a$ and $b$ with solid and, resp., dashed edges and omitted all loops to lighten the picture. The action of $a$ in $\mE_{48}$ is defined by $0\dt a=24\dt a=18$, \ $13\dt a=14$,\  $14\dt a=13$, \  $18\dt a=24$, \ $30\dt a=32$, \ $32\dt a=30$, and $k\dt a=k$ for all other $k\in\Z_{48}$.

\begin{figure}[htb]
\begin{center}
  \begin{tikzpicture}
         \pgfmathsetmacro{\n}{24}
        \foreach \t [evaluate=\t as \teval using int(2*\t)] in {1,...,23} {
        \edef\temp{\noexpand
        \node[fill=white, circle, draw=blue, scale=1] (\teval) at ( {4*cos((360*\t)/\n)}, {4*sin((360*\t)/\n)} )  {\teval};
         }\temp}
         \foreach \t [evaluate=\t as \teval using int(2*\t-1)] in {1,...,\n} {
         \edef\temp{\noexpand
         \node[fill=white, circle, draw=blue, scale=1] (\teval) at ( {5.5*cos((360*\t)/\n-7.5)}, {5.5*sin((360*\t)/\n-7.5)} ) {\teval};
         }\temp}
         \node[fill=white, circle, draw=blue, scale=1] (0) at ( {4}, {0} ) {$0$};
		 \foreach \t [evaluate=\t as \teval using int(\t + 1)]  in {0,1,...,46}{
		 	\draw
		 	(\t) edge[->, dashed ] (\teval);
         }
         \draw
         	(47) edge[->, dashed] (0)
         	(0) edge[->] (18)
         	(18) edge[->, bend left = 60] (24)
         	(24) edge[->, bend right = 15] (18)
         	(14) edge[->, bend right=20] (13)
         	(13) edge[->, bend right=20] (14)
            (30) edge[->, bend right=40] (32)
         	(32) edge[->, bend right=40] (30);
         \end{tikzpicture}
\end{center}
\caption{The DFA $\mE_{48}=\langle\Z_{48},\{a,b\}\rangle$ with $H_2(\mE_{48})\ne\Z_{48}$. Solid and dashed edges show the action of $a$ and, resp., $b$; loops are not shown}\label{fig:e48}
\end{figure}

One can calculate that $D_1(\mE_{48})=\{18,24,42\}$ whence the subgroup $H_1(\mE_{48})$ consists of all residues divisible by 6. Computing $D_2(\mE_{48})$, one sees that this set consists of even residues and contains 2 (due to the word $ab^{32}a$ that has $\excl(ab^{32}a)=\{0,30\}\subseteq H_1(\mE_{48})$ and  $\dupl(ab^{32}a)=\{2,18\}$). Hence the subgroup $H_2(\mE_{48})$ consists of all even residues. Finally, the word $ab^{24}ab^{12}ab^{8}$ has $\{0,8,20\}\subseteq H_1(\mE_{48})$ as its excluded set while its duplicate set contains 13. Hence $13\in D_3(\mathrsfs{E}_{48})$ and the subgroup $H_3(\mathrsfs{E}_{48})$ coincides with $\Z_{48}$. We conclude that the DFA $\mathrsfs{E}_{48}$ is completely reachable by Proposition~\ref{prop:sufficientnew}.

As mentioned, the subgroups of $(\Z_n,+)$ ordered by inclusion correspond to the divisors of $n$ ordered by division whence for any standardized DFA $\mA$ with $n$ states, the number of different subgroups of the form $H_k(\mA)$ is $O(\log n)$. Therefore, if the subgroup sequence $H_0(\mA)\subseteq H_1(\mA)\subseteq\dots\subseteq H_k(\mA)\subseteq\dots$ strictly grows at each step, then it reaches $\Z_n$ after at most $O(\log n)$ steps, and by   Proposition~\ref{prop:sufficientnew} $\mA$ is a \cran. What happens if the sequence stabilizes earlier? Our next result answers this question.

\begin{proposition}
\label{prop:stabilization}
If for a standardized DFA $\mA=\langle \Z_n,\{a,b\}\rangle$, there exists $\ell$ such that $H_\ell(\mA)=H_{\ell+1}(\mA)\subsetneqq\Z_n$, then $\mA$ is not completely reachable.
\end{proposition}

\begin{proof}
As in the proof of Proposition~\ref{prop:sufficientnew}, we use $D_k$ and $H_k$ instead of $D_k(\mA)$ and, resp., $H_k(\mA)$ in our arguments.

It suffices to prove the following claim:

\noindent\textbf{Claim}: \emph{the equality $H_\ell=H_{\ell+1}$ implies that the subgroup $H_\ell$ is $a$-invariant}.

Indeed, since $H_\ell\subsetneqq\Z_n$, we get a proper $a$-invariant subgroup, and  Proposition~\ref{prop:invsubgroup} then shows that $\mA$ is not completely reachable.

Technically, it is more convenient to show that if $H_\ell=H_{\ell+1}$, then $H_k\dt a\subseteq H_\ell$ for every $k=0,1,\dots,\ell$. We induct on $k$. The base $k=0$ is clear since $H_0=\{0\}$ and $0\dt a=\dupl(a)\in D_1\subseteq H_1\subseteq H_\ell$.

Let $k<\ell$ and assume $H_k\dt a\subseteq H_\ell$; we aim to verify that $p\dt a\in H_\ell$ for every $p\in H_{k+1}$. Since the subgroup $H_{k+1}$ is generated by $D_{k+1}$ and contains $H_{k}$, we may choose a representation of $p$ as the sum
\[
p=q+d_1+\dots+d_m, \quad q\in H_k,\ d_1,\dots,d_m\in D_{k+1}\setminus H_k,
\]
with the least number $m$ of summands from $D_{k+1}\setminus H_k$. We show that $p\dt a\in H_\ell$ by induction on $m$. If $m=0$, we have $p=q\in H_k$ and $p\dt a\in H_\ell$ since $H_k\dt a\subseteq H_\ell$.

If $m>0$, we write $p$ as $p=d_1+s$ where $s=q+d_2+\dots+d_m$. By \eqref{eq:defk}, there exists a word $w\in\{a,b\}^*$ such that $d_1\in\dupl(w)$, $0\in\excl(w)\subseteq H_{k}$ and $|\excl(w)|\le k+1$. Consider the word $wb^sa$. We have $p\dt a=(d_1+s)\dt a=d_1\dt b^sa$, and the equality \eqref{eq:duplprod} gives $p\dt a\in\dupl(wb^sa)$. From the equality \eqref{eq:exclprod}, we get
$\excl(wb^sa)=(\excl(w)+s)\dt a\cup\{0\}$ if $\dupl(a)a^{-1}$ is either contained in or disjoint with $\excl(w)+s$, and $\excl(wb^sa)=\bigl((\excl(w)+s)\setminus\dupl(a)a^{-1}\bigr)\dt a\cup\{0\}$ if $|\dupl(a)a^{-1}\cap(\excl(w)+s)|=1$. In any case, we have the inclusion
\begin{equation}
\label{eq:inclusion}
\excl(wb^sa)\subseteq(\excl(w)+s)\dt a\cup\{0\}
\end{equation}
and the inequality
\begin{equation}
\label{eq:inequality}
|\excl(wb^sa)|\le|(\excl(w)+s))\dt a|+1\le|\excl(w))|+1\le (k+1)+1\le \ell+1.
\end{equation}
For any $t\in\excl(w)\subseteq H_{k}$, the number of summands from $D_{k+1}\setminus H_k$ in the sum $t+s=t+q+d_2+\dots+d_m$ is less than $m$. By the induction assumption, we have $(t+s)\dt a\in H_\ell$. Hence, $(\excl(w)+s)\dt a\subseteq H_\ell$, and since 0 also lies in the subgroup $H_\ell$, we conclude from \eqref{eq:inclusion} that $\excl(wb^sa)\subseteq H_\ell$. From this and the inequality \eqref{eq:inequality}, we see that the word $wb^sa$ satisfies the conditions of the definition of $D_{\ell+1}$ (cf.\ \eqref{eq:defk}) whence every state in $\dupl(wb^sa)$ belongs to $D_{\ell+1}$. We have observed that $p\dt a\in\dupl(wb^sa)$. Hence $p\dt a\in D_{\ell+1}\subseteq H_{\ell+1}$. Since $H_\ell=H_{\ell+1}$, we have $p\dt a\in H_\ell$, as required.
\end{proof}

Now we deduce a criterion for complete reachability of binary automata.

\begin{theorem}
\label{thm:binary}
A binary DFA $\mA$ with $n$ states is completely reachable if and only if either $n=2$ and $\mA$ is the flip-flop or one of the letters of $\mA$ acts as a cyclic permutation of the state set, the other letter has defect $1$, and in the standardized DFA $\langle\Z_n,\{a,b\}\rangle$ syntactically equivalent to $\mA$, no proper subgroup of $(\Z_n,+)$ is $a$-invariant.
\end{theorem}

\begin{proof}
Necessity follows from the reductions in Sect.~\ref{sec:background} and Proposition~\ref{prop:invsubgroup}.

For sufficiency, we can assume that $\mA=\langle\Z_n,\{a,b\}\rangle$ is standardized. If no proper subgroup of $(\Z_n,+)$ is $a$-invariant, then the claim from the proof of Proposition~\ref{prop:stabilization} implies that the sequence $H_0(\mA)\subseteq H_1(\mA)\subseteq\dots\subseteq H_k(\mA)\subseteq\dots$ strictly grows as long as the subgroup $H_k(\mA)$ remains proper. Hence, $H_\ell(\mA)=\Z_n$ for some $\ell$ and $\mA$ is a \cran\ by Proposition~\ref{prop:sufficientnew}.
\end{proof}

\begin{remarknew}
\label{rem:h1}
The proof of Theorem \ref{thm:binary} shows that only subgroups that contain $H_1(\mA)$ matter. Therefore, one can combine Theorem~\ref{thm:sufficient}, Proposition~\ref{prop:cayley} and Theorem \ref{thm:binary} as follows:
\emph{a standardized DFA $\mA=\langle\Z_n,\{a,b\}\rangle$ is completely reachable if and only if either $H_1(\mA)=\Z_n$ or no proper subgroup of $(\Z_n,+)$ containing the subgroup $H_1(\mA)$ is $a$-invariant}.
\end{remarknew}

The condition of Theorem \ref{thm:binary} can be verified in low polynomial time. We sketch the corresponding algorithm.

Given a binary DFA $\mA$ with $n$ states, we first check if $n=2$ and $\mA$ is the flip-flop. If \textbf{yes}, $\mA$ is completely reachable. If \textbf{not}, we check whether one of the letters of $\mA$ acts as a cyclic permutation of the state set while the other letter has defect 1. If \textbf{not}, $\mA$ is not completely reachable. If \textbf{yes}, we pass to the standardized DFA $\langle\Z_n,\{a,b\}\rangle$ syntactically equivalent to $\mA$. As a preprocessing, we compute and store the set $\{(k,k\dt a)\mid k\in\Z_n\}$.

The rest of the algorithm can be stated in purely arithmetical terms. Call a positive integer $d$ a \emph{nontrivial divisor} of $n$ if $d$ divides $n$ and $d\ne 1,n$. We compute all nontrivial divisors of $n$ by checking through all integers $d=2,\dots,\lfloor\sqrt{n}\rfloor$: if such $d$ divides $n$, we store $d$ and $\frac{n}d$. If for some nontrivial divisor $d$ of $n$, all numbers $(td)\dt a$ with $t=0,1,\dots,\frac{n}d-1$ are divisible by $d$, then $d$ generates a proper $a$-invariant subgroup in $(\Z_n,+)$ and $\mA$ is not completely reachable. If for every nontrivial divisor $d$ of $n$, there exists $t\in\{0,1,\dots,\frac{n}d-1\}$ such that $(td)\dt a$ is not divisible by $d$, then no proper subgroup of $(\Z_n,+)$ is $a$-invariant and $\mA$ is completely reachable.

To estimate the time complexity of the described procedure, observe that one has to check at most $\frac{n}d$ numbers for each nontrivial divisor $d$ of~$n$. Clearly,
\[
\sum_{\substack{1<d<n\\d|n}}\frac{n}d=\sum_{\substack{1<d<n\\d|n}}d=\sigma(n)-n-1,
\]
where $\sigma(n)$ stands for the sum of all divisors of $n$, a well-studied function in the theory of numbers; see, e.g., \cite[Chapters XVI--XVIII]{Hardy:2008}. It is known that $\limsup\frac{\sigma(n)}{n\log\log n}=e^\gamma$ where $\gamma$ is the Euler--Mascheroni constant \cite[Theorem 323]{Hardy:2008}; this implies that the number of checks in our procedure is $O(n\log\log n)$. The total complexity depends on the time spent for verifying the divisibility condition. If one uses the transdichotomous model~\cite{FrWi93} (as suggested by one of the referees), assuming constant time for division, the whole procedure can be implemented in $O(n\log\log n)$ time.

One can speed up the above algorithm, using Remark~\ref{rem:h1}, which implies that only the divisors $d>1$ of the g.c.d. of $n$ and $0\dt a$ have to be checked. However, the improvement only reduces the constant behind the $O(\ )$ notation.

\section{Conclusion}

We have characterized binary \cra; our characterization leads to an algorithm that given a binary DFA $\mA$, decides whether or not $\mA$ is completely reachable in quasilinear time with respect to the size of $\mA$. Very recently, after the original version of the present paper was submitted, Ferens and Szyku\l{}a~\cite{FeSz22} have devised a polynomial-time algorithm for recognizing complete reachability of arbitrary DFAs, but the complexity of their algorithm is higher.

Our results heavily depend on the fact that apart from a single exception, binary \cra\ are \emph{circular}, that is, have a letter acting as a cyclic permutation of the state set. In the literature, one can find several situations when a problem that remains open in general, admits quite a nontrivial solution when restricted to circular automata. Here we mention only Dubuc's result \cite{Dubuc98} on the \v{C}ern\'{y} conjecture and the recent paper by Yong He \emph{et al} \cite{He2021} on Trahtman's conjecture. It appears that circular automata may behave in a similar way with respect to complete reachability, and our follow-up work aims at extending the results of the present paper to arbitrary (not necessarily binary) circular automata. We also plan to study an `orthogonal' extension, aiming to characterize \cra\ in which one letter has defect~1 while the other letters act as permutations and generate a group that transitively acts on the state set.

\medskip

\noindent\textbf{Acknowledgement.} We thank the anonymous reviewers of the conference version of our paper for their careful reading and their many useful comments and suggestions that are incorporated in the present version.

\end{document}